\documentclass[a4paper,11pt]{article}
\usepackage[utf8]{inputenc}
\usepackage{amsthm,amsmath}
\usepackage{hyperref}

\title{Lower Bounds on the mim-width of Some Graph Classes}
\author{Stefan Mengel\\CNRS CRIL UMR 8188}

\begin{document}
\maketitle

\begin{abstract}
mim-width is a recent graph width measure that has seen applications in graph algorithms and problems related to propositional satisfiability.
In this paper, we show linear lower bounds for the mim-width of strongly chordal split graphs, co-comparability graphs and circle graphs. This improves and refines lower bounds that were known before, some of them only conditionally. In the case of strongly chordal graphs not even a conditional lower bound was known before. All of the bounds given are optimal up to constants.
\end{abstract}

\newtheorem{theorem}{Theorem}
\newtheorem{lemma}[theorem]{Lemma}
\newtheorem{observation}[theorem]{Observation}

\newcommand{\mimw}{\mathsf{mimw}}
\newcommand{\tw}{\mathsf{tw}}

\section{Introduction}

The study of restricted graph classes like chordal graphs, interval graphs and circle graphs is a classical field in graph theory. In particular, these classes have also been studied in the field of graph algorithms because many classical problems that are $\mathsf{NP}$-complete on general graphs are tractable for these classes, see~\cite{isgci} for an overview and extensive references. 

One way of explaining favorable algorithmic properties of certain graph classes is showing that on these classes width measures that can be used for algorithm design are bounded. One such parameter is \emph{mim-width}, short for maximum induced matching width, which has been explicitly defined by Vatshelle in his thesis~\cite{vatshelleThesis} but was implicitly already used by Belmonte and Vatshelle in~\cite{BelmonteV13}. mim-width strictly generalizes many other graph parameters such as treewidth and cliquewidth but still allows efficient algorithms for many graph problems. For example, Belmonte and Vatshelle showed that interval graphs and many other classes have constant mim-width which results directly in polynomial time algorithms for Maximum Independent Set, Minimum Dominating Set and many other problems. Note that these graphs have unbounded cliquewidth, so these algorithms rely on the enhanced power of mim-width.

Another field in which mim-width has had quite some impact are problems related to propositional satisfiability, in particular model counting and knowledge compilation. Here, mim-width cannot only be used for upper bounds, using also his close cousin ps-width~\cite{SaetherTV15,BovaCMS15,GaspersPST16}, but it is also an important ingredient (sometimes slightly modified) when showing lower bounds~\cite{Razgon14,BovaCMS14}.  

It is known that mim-width cannot be used to design algorithms for some well-known graph classes, e.g.~chordal graphs. This is because having small width and being able to compute the relevant decomposition tree would result in efficient algorithms on these graphs for problems that remain $\mathsf{NP}$-complete on them. Vatshelle in his thesis~\cite{vatshelleThesis} gave several classes for which this is the case. Consequently, either the mim-width of these graph classes is generally high or it is small but computing the witnessing decompositions is intractable. In this paper, we show that for all of the classes that Vatshelle gave, the former is the case. More specifically, we show that the class of strongly chordal split graphs, co-comparability graphs and circle graphs generally all have mim-width  bounded from below linearly in the number of vertices. Up to constant factors, these bounds are optimal. Note that in the case of strongly chordal graphs not even conditional lower bounds were known. In fact, Vatshelle in his thesis~\cite{vatshelleThesis} asks if the mim-width of this graph class  if bounded. Our lower bound answers this question negatively.



We remark that very recently Kang et al.~\cite{KangKST16} showed a mim-width lower bound of $\sqrt{\log_2(n/2)}$ for split graphs and of $\sqrt{n/12}$ for co-comparability graphs. Our bounds here beat these bounds quantitatively and, in the case of split graphs, are true for more restrictive graph classes.
We also remark that our work heavily relies on the mim-width lower bound in~\cite{Brault-BaronCM15} for chordal bipartite graphs.

\section{Preliminaries}

We use standard notation from graph theory which can be found in e.g.the textbook~\cite{Diestel}. Let $G=(V,E)$ be a graph. For every set $A\subseteq V$ we define $\bar A := V\setminus A$. Moreover, we define $G[A,\bar A]$ to be the bipartite graph with color classes $A$ and $\bar A$ that contains exactly the edges in $E$ that have endpoints in both $A$ and $\bar A$.

A \emph{branch decomposition} of a graph $G=(V,E)$ is a rooted binary tree~$T$ whose leaves are the vertices in $V$. There are several graph width measures defined by branch decompositions; here we will only be interested in \emph{mim-width} which is defined as follows: Every node $t$ of the branch decomposition~$T$ of $G$ defines a partition $A, \bar A$ of $V$ where $A$ contains exactly the vertices appearing as leaves in the subtree $T_t$ of $T$ rooted in $t$. The mim-value of the partition $A,\bar A$ is the size of a maximum induced matching in $G[A,\bar A]$. The mim-width of the decomposition $T$ is the maximum mim-value of the partitions defined by the nodes $t$ of $T$. Finally, the mim-width of $G$, denoted by $\mimw(G)$, is the minimum mim-width taken over all branch decompositions of $G$.

Let $C$ be a cycle in a graph. We say that $C$ has a \emph{chord} if there is an edge between two vertices of $C$ that are not adjacent on $C$. A graph is called \emph{chordal bipartite} if it is bipartite and every cycle of length at least $6$ has a chord. Chordal bipartite graphs are known to have high mim-width.

\begin{theorem}[\cite{Brault-BaronCM15}]\label{thm:chordalbip}
 There is an infinite family $\mathcal{G}_{cb}$ of chordal bipartite graphs and a constant $c$ such that for every graph $G\in \mathcal G_{cb}$ we have $\mimw(G)\ge c|V(G)|$.
\end{theorem}

\section{Adding Edges to Bipartite Graphs}

In this section we will prove the main technical lemma which we will use to prove lower bounds for restricted graph classes later on. Intuitively, it says that when starting with a bipartite graph, adding edges inside the color classes does not decrease the mim-width by much.

\begin{lemma}\label{lem:addedgeslower}
 Let $G$ be a bipartite graph with color classes $X$ and $Y$. Let moreover $G'$ be a graph that we get from $G$ by adding some edges whose respective end vertices are either both in $X$ or both in $Y$. Then \[\mimw(G') \ge \frac{1}{2}\mimw(G).\]
\end{lemma}
\begin{proof}
Note first that any branch decomposition of $G$ is a branch decomposition of $G'$ and vice versa. Let $T$ be any such branch
decomposition. Then we know that there is a node $t$ in $T$ that defines a partition $(A,\bar{A})$ of the vertices such that  $G[A, \bar{A}]$ has an induced matching $M$ of size $\mimw(G)$. Then $M$ has a sub-matching $M'$ such that
\begin{itemize}
\item the end vertices of $X$ appearing as end vertices in $M'$, call these $X'$, are all in one color class in $G[A,\bar{A}]$, and
\item $|M'| \ge |M|/2 \ge \mimw(G)/2$.
\end{itemize}
Then the vertices $Y$ appearing as end vertices in $M$ are all in the opposite color class of $G[A,\bar{A}]$ as the vertices in $X'$. Consequently, none of the edges added in the construction of $G'$ connect any vertices in $M'$ in $G'[A, \bar A]$. Thus~$M'$ is an induced matching also in $G'[A, \bar A]$ of size at least $\mimw(G)/2$. Since we can find such an $M'$ for every branch decomposition $T$, the claim follows.
\end{proof}

Let us sketch how we can use Lemma~\ref{lem:addedgeslower} to show lower bounds. The idea is to start with a class of bipartite graphs for which we already have a lower bound such as that of Theorem~\ref{thm:chordalbip}. We then construct new graphs by adding edges to the graphs in the class we started from. These edges are added only inside the respective color classes and in such a way that the resulting graphs lie in our target class. Then invoking Lemma~\ref{lem:addedgeslower} tells us that the mim-width has not decreased much in this transformation which completes the proof. 

As an example, we could start with the class of grid graphs which are bipartite and known to have a mim-width of $\Theta(\sqrt{n})$ where $n$ is the number of vertices~\cite[Theorem 4.3.10]{vatshelleThesis}. Then adding all potential edges inside one of the respective color classes yields split graphs (see Section~\ref{sct:split}). With Lemma~\ref{lem:addedgeslower} this yields a lower bound of $\Omega(\sqrt{n})$, beating the $\sqrt{\log_2(n/2)}$ lower bound of~\cite{KangKST16} exponentially.

Our refined lower bounds take a little more care but are shown with essentially the same argument.

\section{Strongly Chordal Split Graphs}\label{sct:split}

A cycle in a graph is called \emph{even} if it contains an even number of vertices. An odd chord in a cycle is a chord that connects two vertices that have an odd distance along the cycle. A graph is defined to be \emph{chordal} if every cycle of length at least $4$ has a chord. A graph is called \emph{strongly chordal} if it is chordal and every even cycle of length at least $6$ in it has an odd chord. A graph is called a \emph{split graph} if its vertices can be partitioned into two sets in such a way that they induce a clique and an independent set, respectively. Finally, a graph is called a \emph{strongly chordal split graph} if it is both strongly chordal and a split graph.

Let $G$ be a bipartite graph with color classes $X$ and $Y$. Construct $G'$ by adding to $G$ an edge between every pair of vertices in $Y$.

\begin{lemma}\label{lem:stronglychordalsplit}
 If $G$ is chordal bipartite, then $G'$ is a strongly chordal split graph.
\end{lemma}
We remark that Lemma~\ref{lem:stronglychordalsplit} was already observed in \cite{Muller96a} without proof.

\begin{proof}
 We first show that $G'$ is a split graph. To see this, observe that by construction $Y$ induces a clique in $G'$ and $X$ induces an independent set, because $G$ is bipartite.
 
 We now show that $G'$ is strongly chordal. First observe that it is chordal, because every cycle of length at least $4$ must contain two vertices in $Y$ that are not neighbors on the cycle. But these vertices are then connected by an edge that was added in the construction of $G'$ from $G$, which gives us the desired chord.
 
 Now consider an even cycle $C$ of length at least $6$ in $G'$. Let $S$ be the set of edges added in the construction of $G'$ from $G$. We consider two cases depending on if $C$ contains edges from $S$ and show that in both cases $C$ has an odd chord and thus $G'$ is strongly chordal.
 
 First assume that $C$ contains no edges from $S$. Thus $C$ is also a cycle in $G$ and has a chord $e$ because $G$ is chordal bipartite. But since $G$ is bipartite, the chord $e$ is odd.
 
 Now assume that $C=c_1c_2\ldots c_m$ contains an edge from $S$, i.e., an edge between two vertices in $Y$, say $c_1c_2$. If there is a vertex $c_k$ in $Y$ where $4\le k\le m-1$, then both $c_1$ and $c_2$ are adjacent to $c_k$. Since $c_1$ and $c_2$ are adjacent on $C$, one of the chords $c_1c_k$ or $c_2c_k$ must be odd and we are done. Thus, we may assume that $V(C)\cap Y\subseteq \{c_m, c_1, c_2, c_3\}$, and they appear consecutively.
 Thus, $|V(C)|\le 5$ as $X$ is an independent set. This is a contradiction to the assumption that $|V(C)|\ge 6$.
\end{proof}

We now have everything in place to show our first mim-width lower bound.

\begin{theorem}\label{thm:scs}
  There is an infinite family $\mathcal{G}_{scs}$ of strongly chordal split graphs and a constant $c$ such that for every graph $G\in \mathcal G_{scs}$ we have $\mimw(G)\ge c|V(G)|$.
\end{theorem}
\begin{proof}
 Consider the family $\mathcal{G}_{cb}$ from Theorem~\ref{thm:chordalbip} and compute for every graph~$G$ in it the graph $G'$ as above. The union of the $G'$ is the family~$\mathcal{G}_{scs}$. By Lemma~\ref{lem:stronglychordalsplit}, the family $\mathcal{G}_{scs}$ consists of strongly chordal split graphs. Lemma~\ref{lem:addedgeslower} shows that in the construction of $G'$ from $G$, the mim-width only decreases by a constant factor, which directly yields the lower bound for $\mathcal{G}_{scs}$ from that for $\mathcal{G}_{cb}$.
\end{proof}

We remark that Theorem~\ref{thm:scs} answers negatively half of Open Questions~9 from~\cite{vatshelleThesis} which asks if the mim-width of strongly chordal graphs and tolerance graphs is bounded by a constant.

\section{Co-Comparability Graphs}

A graph G is a \emph{comparability graph} if it is transitively orientable, i.e., its edges can be directed such that if $ab$ and $bc$ are directed edges, then $ac$ is a directed edge as well. Clearly, all bipartite graphs are comparability graphs: just choose one color class and direct all edges from this color class to the other one. Since there are no directed paths containing more than one edge in the resulting digraph, the transitivity condition is trivially true. A graph is called a co-comparability graph if it is the complement of a comparability graph.

For every bipartite graph $G$ with color classes $X$ and $Y$, we construct a graph~$G'$ by adding edges between all pairs of vertices that lie in the same color class.

\begin{observation}
 The graph $G'$ is a co-comparability graph.
\end{observation}
\begin{proof}
 Since all pairs of vertices in $X$, respectively $Y$, are connected in $G'$, in the complement graph $\bar{G'}$ both $X$ and $Y$ induce independent sets. Thus $\bar{G'}$ is bipartite and consequently a comparability graph. It follows that $G'$ is a co-comparability graph.
\end{proof}

The following theorem can now be proved with essentially the same proof as Theorem~\ref{thm:scs}.

\begin{theorem}\label{thm:cc}
  There is an infinite family $\mathcal{G}_{cc}$ of co-comparability graphs and a constant $c$ such that for every graph $G\in \mathcal G_{cc}$ we have $\mimw(G)\ge c|V(G)|$.
\end{theorem}

With Theorem~\ref{thm:chordalbip} and Theorem~\ref{thm:cc}, we know that there is an infinite family of comparability graphs and an infinite family of co-comparability graphs that both have high mim-width. Let us remark that in contrast to this all families of graphs that are comparability graphs \emph{and} co-comparability graphs have constant mim-width. Indeed, such graphs are circular permutation graphs (see~\cite{isgci}) which by~\cite{BelmonteV13} have mim-width at most $2$.

\section{Circle Graphs}

In this section, we will show that the class of circle graphs has high mim-width, using a construction that is heavily inspired by that of~\cite{Damaschke89}. A graph $G$ is called a circle graph if it is the intersection graph of chords in a circle, i.e., it has the following representation: the vertices are chords in a circle and two vertices are connected by an edge if and only if these chords intersect.

We will start with the following lemma.

\begin{lemma}\label{lem:circhelp}
 There is an infinite family $\mathcal G$ of bipartite graphs and a constant $c>0$ such that the following statements are true:
 \begin{itemize}
  \item[1)] Let $G$ in $\mathcal G$ have color classes $X$ and $Y$. Then all vertices in $X$ have degree $3$ while all vertices in $Y$ have degree $2$.
  \item[2)] Every graph $G$ in $\mathcal G$ has $\mimw(G)\ge c|V(G)|$.
 \end{itemize}
\end{lemma}
\begin{proof}
 Let $\mathcal G^*$ be an infinite family of $3$-regular graphs such that there is a constant $c'$ such that $\tw(G)\ge c'|V(G)|$ for every graph~$G$ from $\mathcal G^*$. Such families of graphs were shown to exist in~\cite{GroheM09}. For $G$ from $\mathcal G^*$ we construct~$G'$ by substituting every edge by a path of length $2$. We define $\mathcal G$ to be the union of all the $G'$.
 
 By construction, the graphs in $\mathcal G$ satisfy property 1), so it only remains to show property 2). Remember that a graph is called \emph{$d$-degenerate} if every subgraph has a vertex of degree at most $d$. Note that $G'$ as constructed above is $2$-degenerate. Combining Lemma 4.2.4 and Lemma 4.3.9 from~\cite{vatshelleThesis}, we get that for every $d$-degenerate graph we have 
 \begin{align} \mimw(G) \ge \frac{\tw(G)}{3(d+1)}.\end{align}
 
 Remember that the graphs in~$\mathcal G^*$ have treewidth linear in the number of vertices. Since for all graphs in $\mathcal G^*$ the number of edges and thus the number of vertices of all graphs in $\mathcal G$ is linear in the number of vertices of the graphs in $\mathcal G^*$ that were used in the construction of $\mathcal G$, we get that there is a constant~$c$ satisfying property~2).
\end{proof}

We now show that circle graphs have high mim-width.

\begin{theorem}\label{thm:circle}
 There is an infinite family $\mathcal{G}_{cir}$ of circle graphs and a constant $c$ such that for every graph $G\in \mathcal G_{cir}$ we have $\mimw(G)\ge c|V(G)|$. 
\end{theorem}
\begin{proof}
 Let $\mathcal G$ be the class of graphs in Lemma~\ref{lem:circhelp}. For every $G$ in $\mathcal G$ we construct a circle graph $G'$ as follows: For every vertex in $X$, we put a chord into a circle which we call an $X$-chord. We do this in such a way that all X-chords have the same length and no pair of them intersect. Note that by doing so for every $X$-chord one of the arcs defined by its end points contains no endpoint of any of the other $X$-chords. Now for every vertex in $Y$ we add a chord in such a way that it intersects the $X$-chords of its neighbors and no other $X$-chord. This completes the construction of $G'$. Note that the construction of $G'$ is not deterministic as there are in general several concrete ways to put the chord into the circle, but for our purposes any graph constructed this way will do, so for every $G$ from $\mathcal G$ we arbitrarily pick one such graph and define it to be $G'$. We define $\mathcal{G}_{cir}$ to be the union of all $G'$ constructed this way.
 
 By construction, $\mathcal{G}_{cir}$ consists of circle graphs, so it only remains to show the lower bound on the mim-width. To this end, observe that $G$ is a subgraph of $G'$ that we get by deleting all edges that connect pairs of vertices in $Y$. Said differently, $G'$ can be constructed from $G$ by adding some edges between vertices in $Y$. Lemma~\ref{lem:circhelp} and Lemma~\ref{lem:addedgeslower} directly yield the desired lower bound.
\end{proof}

\section{Conclusion}

We showed a mim-width lower bound linear in the number of vertices for several well known graph classes. This strengthens conditional lower bounds of~\cite{vatshelleThesis} and recent bounds from~\cite{KangKST16}. We also answered half of an open question by Vatshelle, showing that strongly chordal graphs do not have bounded mim-width. It would be interesting to answer the other half of this question either by showing that tolerance graphs have bounded mim-width or by constructing tolerance graphs of high mim-width.

As we have seen, Lemma~\ref{lem:addedgeslower} is quite a flexible tool in proving mim-width lower bounds. We are confident that it or similar arguments can be used to show more lower bounds in a similar way.

Finally, it would be interesting to perform a similar study as in this paper for the very recently introduced notion of sim-width, a refinement of mim-width~\cite{KangKST16}. Since split graphs have sim-width $1$ but linear mim-width, the gap between the two measures is essentially as big as it could potentially be. Can one better understand the relation of sim-width and mim-width for other graph classes? Note that by the results of~\cite{KangKST16}, this would probably require a fine understanding of induced minors isomorphic to so-called t-matching complete graphs.

\paragraph*{Acknowledgements}
The author would like to thank Martin Vatshelle for introducing him to the questions considered in this paper and for helpful discussions. Moreover, the author is grateful for numerous corrections to the paper given by the anonymous reviewers and a substantial simplification of the proof of Lemma~\ref{lem:stronglychordalsplit} proposed by one of the reviewers.

\bibliographystyle{alpha}
\bibliography{mimwidth}
\end{document}